\documentclass{CSML}

\def\dOi{11(4:15)2015}
\lmcsheading%
{\dOi}
{1--9}
{}
{}
{Feb.~23, 2015}
{Dec.~22, 2015}
{}

\ACMCCS{[{\bf Software and its engineering}]: Software organization
  and properties---Software functional properties---Formal methods;
  Software organization and properties---Software system
  structures---Software system models---Petri nets} 

\usepackage{hyperref}

\usepackage{macros}

\begin{document}

\title[Structurally Cyclic Petri Nets]{Structurally Cyclic Petri Nets}

\author[F.~Drewes]{Frank Drewes\rsuper a}	
\address{{\lsuper a}Dept.\ of Computing Science, Ume\aa\ University, Ume\aa, Sweden}	
\email{drewes@cs.umu.se}  

\author[J.~Leroux]{J\'er\^ome Leroux\rsuper b}	
\address{{\lsuper b}LaBRI, CNRS, Univ. Bordeaux, Talence, France}	
\email{leroux@labri.fr}  





\begin{abstract}
  \noindent A Petri net is structurally cyclic if every configuration is reachable
from itself in one or more steps. We show that structural cyclicity is
decidable in deterministic polynomial time. For this, we adapt the
Kosaraju's approach for the general reachability problem for Petri
nets.
\end{abstract}

\maketitle


\section{Introduction}
\label{sec:intro}
Reachability problems for Petri nets are not only famously difficult and computationally complex, but also important from an application point of view. Therefore, reachability has attracted a lot of attention. Three decades ago, the reachability problem for general Petri nets was shown to be decidable by Mayr and Kosaraju~\cite{Mayr:1984,Kos1982}, but to date no primitive recursive upper bound on its complexity is known.

One of the many papers in which variants of the problem are studied is~\cite{Leroux:2013}. There, the stronger property of reversible reachability is shown to be EXPSPACE complete. The reversible reachability problem consists in deciding if two configurations are in the same strongly connected component of the reachability graph.

A natural special case of reversible reachability is the question whether a given configuration $\vec c$ is \emph{cyclic}, i.e., whether it is reachable from itself by one or more steps. In the present paper, we show first that this problem is EXPSPACE complete as well. Then we move on to the main topic of this paper, namely the problem of \emph{structural cyclicity}. A Petri net $T$ is said to be structurally cyclic if each of its configurations is cyclic. Equivalently, $T$ is structurally cyclic if the zero configuration is reachable from itself in $T$ (by at least one step). We show that structural cyclicity can be decided in deterministic polynomial time. This is achieved by studying the set of \emph{markable indices} of $T$, i.e., those indices which, starting from the zero configuration, can be made non-zero on both forward and backward firing sequences, and the set of \emph{ultimately cyclic transitions} of $T$, i.e. transitions that occurs on a cyclic execution.

Apart from the fact that structural cyclicity seems to be a rather natural property, motivation for this work is provided by its usefulness in other areas. In fact, the questions answered in this paper were raised by ongoing work on a basic type of DAG automata in~\cite{Chiang.etAl:2015}. Let us briefly explain this connection. A DAG is a directed acyclic graph with node labels taken from a finite alphabet. A DAG automaton $A$ has a finite set of states and rules of the form $\{p_1,\dots,p_m\}\mathrel{\mathop{\rightarrow}\limits^a}\{q_1,\dots,q_n\}$, where $a$ is a node label and $\{p_1,\dots,p_m\}$ and $\{q_1,\dots,q_n\}$ are multisets of states. A run of $A$ is any assignment of states to the edges of the DAG; such a run is accepting if it is locally consistent with the rules. In other words, for each node, the label of this node together with the multisets of states on its incoming and outgoing edges must form a rule of the DAG automaton. The DAG language $L(A)$ accepted by $A$ is the set of all nonempty DAGs $D$ such that there exists an accepting run of $A$ on $D$. (Note that only nonempty DAGs are considered, because the empty DAG would always be accepted according to these definitions.)

Now, since DAGs are acyclic, a run can be considered as a top-down process that starts at the roots of the DAG and applies rules until it reaches the leaves. Changing perspective slightly, this can be used to view $A$ as a generating device that starts with an empty DAG. In each step, it applies a rule as above by taking $m$ ``dangling'' edges that carry states $p_1,\dots,p_m$, making them the incoming edges of a new node labelled $a$, and adding $n$ dangling outgoing edges to this node, which carry the states $q_1,\dots,q_n$. The process may stop whenever a DAG is obtained that does not contain any further dangling edges. Note that, since the DAG is empty at the very beginning, and thus there are no dangling edges, at least one rule of the form $\emptyset\mathrel{\mathop{\rightarrow}\limits^a}\{q_1,\dots,q_n\}$ must be applied to produce a root (and $n$ dangling edges). Likewise, termination requires the application of rules of the form $\{p_1,\dots,p_m\}\mathrel{\mathop{\rightarrow}\limits^a}\emptyset$ that produce leaves.

Now, by viewing states as dimensions (or places) of a Petri net and adding a transition for each rule of a DAG automaton $A$, one gets a Petri net $T$ which mimics the production and consumption of (states on) dangling edges. In particular, $T$ can turn the zero configuration (corresponding to the start, in which no states are available) into the zero configuration (now corresponding to a terminal situation in which all states have been consumed) if and only if at least one (nonempty) DAG is accepted by $A$. In other words, $T$ is structurally cyclic if and only if the $L(A)$ is nonempty. In this way, our main result shows that the emptiness problem for DAG automata can be solved in deterministic polynomial time. The details of this construction will be found in~\cite{Chiang.etAl:2015}.

\section{Petri Nets}
\label{sec:petrinets}
In the sequel, $d$ denotes a natural number in $\setN$, called the \emph{dimension}. A vector in $\setN^d$ is called \emph{configuration}. Configurations are ordered pointwise by $\vec{x}\leq\vec{y}$ if $\vec{x}(i)\leq \vec{y}(i)$ for every $1\leq i\leq d$. Given a configuration $\vec{c}$, we denote by $\pos{\vec{c}}$ the set of indexes $i$ in $\{1,\ldots,d\}$ such that $\vec{c}(i)>0$. A \emph{Petri net} is a finite set $T$ of pairs of configurations called \emph{transitions}. In this paper, numbers are encoded in binary. That defines the size of configurations and the size of transitions as the sum of the sizes of each component. The size of a Petri net is defined as the sum of the sizes of its transitions. 

The semantics of a Petri net is given by the binary relations $\xrightarrow{t}$ over configurations: for every transition $t\in T$ of the form $(\vec{u},\vec{v})$, we let $\vec{x}\xrightarrow{t}\vec{y}$ if there exists a configuration $\vec{z}$ such that $\vec{x}=\vec{u}+\vec{z}$ and $\vec{y}=\vec{v}+\vec{z}$, with the sum of two vectors defined componentwise. It follows that $\vec{y}=\vec{x}+\Delta(t)$ where $\Delta(t)=\vec{v}-\vec{u}$ is a vector of integers in $\setZ^d$ called the \emph{displacement} of $t$.

This relation is extended to words $w=t_1\dots t_k$ in $T^*$ (where $t_1,\ldots,t_k\in T$) by letting $\vec{x}\xrightarrow{w}\vec{y}$ if $\vec{x},\vec{y}$ are two configurations such that there exists a sequence $\vec{c}_0,\ldots,\vec{c}_k$ of configurations satisfying
$$\vec{x}=\vec{c}_0\xrightarrow{t_1}\vec{c}_1\cdots \xrightarrow{t_k}\vec{c}_k=\vec{y}\;.$$
It follows that $\vec{y}=\vec{x}+\Delta(w)$ where $\Delta(w)\eqdef\sum_{j=1}^k\Delta(t_j)$ is the \emph{displacement} of $w$. By lifting up configurations $\vec{c}_0,\ldots,\vec{c}_k$ by a vector $\vec{z}$, we deduce the following classical fact:
\begin{ourfact}\label{fact:monotony}
  If $\vec{x}\xrightarrow{w}\vec{y}$ then $(\vec{x}+\vec{z})\xrightarrow{w}(\vec{y}+\vec{z})$ for every configuration $\vec{z}$.
\end{ourfact}
The relation is also extended over the languages $W\subseteq T^*$ by letting $\xrightarrow{W}$ denote $\bigcup_{w\in W}\xrightarrow{w}$.

\medskip

A \emph{configuration} $\vec{c}$ is said to be \emph{cyclic} if $\vec{c}\xrightarrow{T^+}\vec{c}$. In Section~\ref{sec:expspace}, we show that deciding if a configuration is cyclic is EXPSPACE complete. In this paper, we are mainly interested in a structural version of the cyclicity problem. Formally, a Petri net $T$ is said to be \emph{structurally cyclic} if every configuration $\vec{c}$ is cyclic. From Fact~\ref{fact:monotony}, it follows that a Petri net $T$ is structurally cyclic if, and only if, $\vec{0}$ is cyclic for $T$. In the sequel, we provide a deterministic polynomial time algorithm for deciding that problem. Our algorithm is based on the computation of the set $\Lambda(T)$ of transitions $t\in T$ that occur in a word $w\in T^+$ witnessing the structural cyclicity $\vec{0}\xrightarrow{w}\vec{0}$. Notice that $T$ is structurally cyclic if, and only if, $\Lambda(T)$ is nonempty. In order to compute $\Lambda(T)$, we provide two different ways for computing subsets $T'$ of $T$ that over-approximate $\Lambda(T)$, i.e., such that $\Lambda(T)\subseteq T'$. These subsets will be useful for simplifying the computation of $\Lambda(T)$ by observing that for every $T'\subseteq T$ such that $\Lambda(T)\subseteq T'$, we have $\Lambda(T)=\Lambda(T')$.

\medskip

The first over-approximation of $\Lambda(T)$ is obtained by introducing the \emph{markable indexes}. An index $i$ in $\{1,\ldots,d\}$ is said to be \emph{forward markable} for a Petri net $T$ if there exists a configuration $\vec{c}$ such that $\vec{0}\xrightarrow{T^*}\vec{c}$ and $i\in\pos{\vec{c}}$. We denote by $I_+(T)$ the set of indexes forward markable for $T$. Symmetrically, we denote by $I_-(T)$ the set of all $i\in\{1,\ldots,d\}$ that are \emph{backward markable}, i.e., such that there exists $\vec c$ with $\vec{c}\xrightarrow{T^*}\vec{0}$ and $i\in\pos{\vec{c}}$. We denote by $I(T)$ the set $I_+(T)\cap I_-(T)$. A transition $t$ in $T$ of the form $(\vec{u},\vec{v})$ with $\pos{\vec{u}}\cup\pos{\vec{v}}\subseteq I(T)$ is said to be \emph{mutually fireable}. We denote the set of all mutually fireable transitions of $T$ by $M(T)$.
\begin{lemma}\label{lem:positive}
  It holds that $\Lambda(T)\subseteq M(T)$.
\end{lemma}
\begin{proof}
  Let $t=(\vec u,\vec v)$ be a transition in $T$ such that:
  $$\vec{0}\xrightarrow{T^*}\vec{x}\xrightarrow{t}\vec{y}\xrightarrow{T^*}\vec{0}$$
  Observe that $\pos{\vec{x}},\pos{\vec{y}}$ are included in $I(T)$. Moreover since $\vec{x}\xrightarrow{t}\vec{y}$, there exists a configuration $\vec{z}$ such that $\vec{x}=\vec{u}+\vec{z}$ and $\vec{y}=\vec{v}+\vec{z}$. We derive that $\pos{\vec{u}}\subseteq \pos{\vec{x}}\subseteq I(T)$ and $\pos{\vec{v}}\subseteq \pos{\vec{y}}\subseteq I(T)$, which proves the lemma.
\end{proof}

\medskip

The second over-approximation of $\Lambda(T)$ is based on the notion of \emph{ultimate cyclicity}. A transition $t$ in a Petri net $T$ is said to be \emph{ultimately cyclic} if it occurs in a word $w\in T^+$ such that $\vec{c}\xrightarrow{w}\vec{c}$ for some configuration $\vec{c}$. We denote by $U(T)$ the set of ultimately cyclic transitions. By definition, $\Lambda(T)$ is contained in $U(T)$:
\begin{lemma}\label{lem:weak}
  It holds that $\Lambda(T)\subseteq U(T)$.
\end{lemma}

\medskip

In Sections~\ref{sec:positive} and~\ref{sec:weak} the sets $M(T)$ and $U(T)$ are shown to be computable in deterministic polynomial time. In particular, by considering $T'\eqdef M(T)\cap U(T)$, we get an over-approximation of $\Lambda(T)$. If $T'=T$, we prove in Section~\ref{sec:theory} that $\Lambda(T)=T$. Otherwise, since $\Lambda(T)=\Lambda(T')$ we reduce the computation of $\Lambda(T)$ to that of $\Lambda(T')$ where $T'$ is strictly included in $T$. With an immediate induction, we show in Section~\ref{sec:theory} that $\Lambda(T)$ is computable in deterministic polynomial time. This complexity is shown to be optimal in that section up to logspace reductions, i.e., we prove \ptime-hardness of the structural cyclicity problem.

\section{The Cyclicity Problem}
\label{sec:expspace}
The \emph{cyclicity problem} consists in deciding if a configuration $\vec{c}$ in $\setN^d$ is cyclic. This problem takes as input a Petri net $T$ and a configuration $\vec{c}$. The following theorem shows that this problem is decidable in exponential space.
\begin{theorem}
  The cyclicity problem is EXPSPACE complete.
\end{theorem}
\begin{proof}
  The cyclicity problem is shown to be in EXPSPACE thanks to a reduction to the reversible reachability problem. The reversible reachability problem takes as input a triple $(\vec{x},T,\vec{y})$ where $T\subseteq\setN^d\times\setN^d$ is a Petri net, $\vec{x},\vec{y}$ are configurations in $\setN^d$, and it decides if both relations $\vec{x}\xrightarrow{T^*}\vec{y}$ and $\vec{y}\xrightarrow{T^*}\vec{x}$ hold. This problem is known to be EXPSPACE complete when the vectors of the Petri net $T$ and the configurations $\vec{x},\vec{y}$ are encoded in binary~\cite{Leroux:2013}. Let us reduce the cyclicity problem to that problem. We consider a Petri net $T\subseteq \setN^d\times\setN^d$ and a configuration $\vec{x}\in\setN^d$. We introduce the set $\vec{Y}=\{\vec{y}\in\setN^d \mid \vec{x}\xrightarrow{T}\vec{y}\}$. Notice that $\vec{Y}$ contains at most $|T|$ configurations. Moreover, the configuration $\vec{x}$ is cyclic if, and only if, there exists $\vec{y}\in\vec{Y}$ such that $\vec{x}\xrightarrow{T^*}\vec{y}$ and $\vec{y}\xrightarrow{T^*}\vec{x}$. Therefore, the cyclicity problem is decidable in EXPSPACE by reduction to at most $|T|$ instances of the reversible reachability problem.

  The EXPSPACE hardness is proved thanks to a reduction of the reachability problem for lossy Petri nets. A Petri net $T\subseteq \setN^d\times\setN^d$ is said to be \emph{lossy} if $(\vec{e}_i,\vec{0})\in T$ for every $1\leq i\leq d$ where $\vec{e}_i$ is the unit vectors in $\setN^d$ defined by $\vec{e}_i(j)=1$ if $j=i$ and $\vec{e}_i(j)=0$ otherwise. Notice that a lossy Petri net $T$ satisfies $\vec{c}\xrightarrow{T^*}\vec{y}$ for all configurations $\vec{c}\geq \vec{y}$. The reachability problem for lossy Petri nets takes as input a triple $(\vec{x},T,\vec{y})$ where $T\subseteq \setN^d\times\setN^d$ is a lossy Petri net with vectors encoded in binary, and $\vec{x},\vec{y}$ are two configurations in $\setN^d$ encoded in binary as well, and it decides if $\vec{x}\xrightarrow{T^*}\vec{y}$. The reachability problem for lossy Petri nets is known to be EXPSPACE complete~\cite{CardozaLiptonMeyer:1976,Rac1978}. We reduce the reachability problem for lossy Petri nets to the cyclicity problem as follows. Let us consider a lossy Petri net $T\subseteq \setN^d\times\setN^d$ and two configurations $\vec{x},\vec{y}$ in $\setN^d$. The reduction creates a Petri net $S\subseteq \setN^{d+1}\times\setN^{d+1}$ from $T$ by adding one extra dimension. We introduce the mapping $\phi:T\rightarrow\setN^{d+1}\times\setN^{d+1}$ defined by $\phi(\vec{u},\vec{v})=((\vec{u},0),(\vec{v},1))$. This function is extended over the words in $T^*$ by $\phi(t_1\ldots t_k)=\phi(t_1)\ldots\phi(t_k)$. The Petri net $S$ is defined as follows where $s_{\mathrm{down}}\eqdef((\vec{y},1),(\vec{y},0))$, $s_{\mathrm{reset}}\eqdef ((\vec{y},0),(\vec{x},0))$, and $\phi(T)\eqdef\{\phi(t) \mid t\in T\}$:
  $$S\eqdef\{s_\mathrm{down},s_\mathrm{reset}\}\cup \phi(T).$$
  Let us prove that $\vec{x}\xrightarrow{T^*}\vec{y}$ if, and only if, $(\vec{x},0)$ is cyclic for $S$. Notice that if there exists a word $w\in T^*$ such that $\vec{x}\xrightarrow{w}\vec{y}$ then $(\vec{x},0)\xrightarrow{\pi}(\vec{x},0)$ where $\pi\eqdef\phi(w)s_{\mathrm{down}}^{|w|}s_{\mathrm{reset}}$. Thus $(\vec{x},0)$ is cyclic for $S$. Conversely, let us assume that $(\vec{x},0)$ is cyclic for $S$. If $\vec{y}\leq \vec{x}$ then $\vec{x}\xrightarrow{T^*}\vec{y}$ since $T$ is a lossy Petri net. So, we can assume that $\vec{y}\not\leq\vec{x}$. There exists a word $\pi\in S^+$ such that $(\vec{x},0)\xrightarrow{\pi}(\vec{x},0)$. Consider the maximal word $w\in T^*$ such that $\phi(w)$ is a prefix of $\pi$, and let $(\vec{c},n)\in\setN^d\times\setN$ be the configuration such that $(\vec{x},0)\xrightarrow{\phi(w)}(\vec{c},n)$. Notice that $n=|w|$ and $\vec{x}\xrightarrow{w}\vec{c}$. As $\vec{y}\not\leq\vec{x}$, the unique transition in $S$ that can be executed from $(\vec{x},0)$ is a transition in $\phi(T)$. It follows that $|w|\geq 1$. Thus $n\geq 1$. It implies that $(\vec{c},n)\not=(\vec{x},0)$. Thus $\phi(w)$ is a proper prefix of $\pi$. By maximality of $w$, it follows that $\phi(w) s_\mathrm{down}$ or $\phi(w) s_\mathrm{reset}$ is a prefix of $\pi$. In both cases, it implies that $\vec{c}\geq\vec{y}$. As $T$ is lossy, we get $\vec{c}\xrightarrow{T^*}\vec{y}$. Therefore $\vec{x}\xrightarrow{T^*}\vec{y}$. We have reduced the reachability problem for lossy Petri nets to the cyclicity problem. This problem is thus EXPSPACE hard.
\end{proof}

\section{Mutually Fireable Transitions}
\label{sec:positive}
In this section we provide a way for computing in deterministic polynomial time the set $M(T)$ of mutually fireable transitions. The following lemma will provide a way for computing $I_+(T)$, the set of forward markable indexes:
\begin{lemma}\label{lem:tt}
  Let $\vec{y}\in\setN^d$ be such that $\vec{0}\xrightarrow{T^*}\vec{y}$ and let $t=(\vec{u},\vec{v})$ be a transition in $T$ such that $\pos{\vec{u}}\subseteq \pos{\vec{y}}$. Then there exists a configuration $\vec{y}'$ in $\setN^d$ satisfying $\vec{0}\xrightarrow{T^*}\vec{y}'$ and $\pos{\vec{y}'}=\pos{\vec{v}}\cup \pos{\vec{y}}$.
\end{lemma}
\begin{proof}
  Let us consider a word $w$ in $T^*$ such that $\vec{0}\xrightarrow{w}\vec{y}$. By Fact~\ref{fact:monotony}, it follows that $\vec{0}\xrightarrow{w}\vec{y}\xrightarrow{w}2\vec{y}\xrightarrow{w}\cdots \xrightarrow{w}n\vec{y}$ for every $n\in\setN$. Choose $n\geq 1$ such that $\vec{u}(i)< n$ for every $i$ in $\pos{\vec{u}}$. 
  Since $\pos{\vec{u}}\subseteq \pos{\vec{y}}$, it follows that $n\vec{y}(i)\geq n$ for every $i$ in $\pos{\vec{u}}$. Thus, $\vec{z}=n\vec{y}-\vec{u}$ is a vector in $\setN^d$ such that $\vec{z}(i)>0$ for 
  every $i$ in $\pos{\vec{y}}$. We deduce that $n\vec{y}\xrightarrow{t}\vec{z}+\vec{v}$. Hence $\vec{y}'\eqdef\vec{z}+\vec{v}$ satisfies the lemma.
\end{proof}

\newcommand{\propT}{\operatorname{prop}_{T}}
Let us define the mapping $\propT$ over the sets $I\subseteq\{1,\ldots,d\}$ by:
$$\propT(I)=\bigcup_{(\vec{u},\vec{v})\in T,\ \pos{\vec{u}}\subseteq I}\pos{\vec{v}}$$
Since this mapping is monotonic for the inclusion relation $\subseteq$, it has a unique minimal fixpoint $I$ with respect to inclusion, i.e., $I$ is the minimal set such that $\propT(I)=I$. This fixpoint can be computed in deterministic polynomial time with a Kleene iteration in at most $d$ steps starting from $I_0=\emptyset$, and the induction $I_k=\propT(I_{k-1})$. The following lemma shows that this fixpoint is the set $I_+(T)$:
\begin{lemma}
  The minimal fixpoint of $\propT$ is $I_+(T)$.
\end{lemma}
\begin{proof}
  By induction, from Lemma~\ref{lem:tt} we derive that for every $k\in\setN$ there exists $\vec{y}_k$ such that $\vec{0}\xrightarrow{T^*}\vec{y}_k$ and $\pos{\vec{y}_k}=I_k$. Thus $\bigcup_k I_k\subseteq I_+(T)$. Conversely, let $i\in I_+(T)$. There exists a configuration $\vec{y}$ in $\setN^d$ and a word $w=t_1\dots t_k\in T^*$ such that $\vec{0}\xrightarrow{w}\vec{y}$ and $\vec{y}(i)>0$. Let $t_j=(\vec u_j,\vec v_j)$ for $j\in\{1,\dots,k\}$, and consider the sequence of configurations $\vec{c}_0,\ldots,\vec{c}_k$ in $\setN^d$ such that:
  $$\vec 0=\vec{c}_0\xrightarrow{t_1}\vec{c}_1\cdots \xrightarrow{t_k}\vec{c}_k=\vec{y}\;.$$
  Observe that $\pos{\vec{c}_0}=\emptyset=I_0$. Assume by induction that $\pos{\vec{c}_{j-1}}\subseteq I_{j-1}$ for some $j\leq k$, and let us prove that $\pos{\vec{c}_j}\subseteq I_j$. Since $\vec{c}_{j-1}\xrightarrow{t_j}\vec{c}_j$, we deduce that $\pos{\vec{u}_j}\subseteq \pos{\vec{c}_{j-1}}\subseteq I_{j-1}$. Thus $\pos{\vec{v}_j}\subseteq I_j$ since $I_j=\propT(I_{j-1})$. In particular, we have proved that $\pos{\vec{c}_k}\subseteq I_k$. Since $\vec{y}=\vec{c}_k$ and $\vec{v}(i)>0$, we deduce that $i\in I_k$. Hence, $I_+(T)\subseteq \bigcup_k I_k$.
\end{proof}

We deduce from the preceding lemma that $I_+(T)$ is computable in deterministic polynomial time. Moreover, the two previous lemmas show that there exists a configuration $\vec{y}$ in $\setN^d$ such that $\vec{0}\xrightarrow{T^*}\vec{y}$ and $\pos{\vec{y}}=I_+(T)$. For the backward case, just observe that $I_-(T)=I_+(T^{-1})$ where $T^{-1}\eqdef\{(\vec{v},\vec{u})\mid (\vec{u},\vec{v})\in T\}$. Thus, we have proved the following theorem.

\begin{theorem}\label{thm:positive}
  The set $M(T)$ of mutually fireable transitions is computable in deterministic polynomial time. Moreover, if every transition is mutually fireable, there exist configurations $\vec{x},\vec{y}$ in $\setN^d$ such that $\pos{\vec{x}}=I(T)=\pos{\vec{y}}$, and such that:
  $$\vec{y}\xrightarrow{T^*}\vec{0}\xrightarrow{T^*}\vec{x}\;.$$
\end{theorem}
\begin{proof}
  Since $I_+(T)$ and $I_-(T)$ are computable in deterministic polynomial time, the sets $I(T)$ and $M(T)$ are computable with the same complexity. Now, assume that every transition is mutually fireable. We have proved that there exists $\vec{x},\vec{y}$ in $\setN^d$ such that $\pos{\vec{x}}=I_+(T)$, $\pos{\vec{y}}=I_-(T)$ and such that:
  $$\vec{y}\xrightarrow{T^*}\vec{0}\xrightarrow{T^*}\vec{x}$$
  Since every transition $(\vec{u},\vec{v})\in T$ satisfies $\pos{\vec{v}}\subseteq I(T)$, we deduce that $\pos{\vec{x}}\subseteq I(T)$. From this and the inclusion $I(T)\subseteq I_+(T)=\pos{\vec{x}}$, we deduce the equality $\pos{\vec{x}}=I(T)$. Symmetrically, we get $\pos{\vec{y}}=I(T)$.
\end{proof}

\section{Ultimately Cyclic Transitions}
\label{sec:weak}
In this section, the set $U(T)$ of ultimately cyclic transitions is shown to be computable in polynomial time. The \emph{displacement} of a function $\psi\colon T\rightarrow\setN$ is the vector in $\setZ^d$ defined by $\Delta(\psi)\eqdef\sum_{t\in T}\psi(t)\Delta(t)$. The following theorem follows quite immediately from linear algebra:
\begin{theorem}\label{thm:weak}
  The set $U(T)$ is computable in deterministic polynomial time. Moreover, there exists $\psi\colon T\rightarrow\setN$ such that $\Delta(\psi)=\vec{0}$ and $\psi(t)\geq 1$ for all $t\in U(T)$.
\end{theorem}
\begin{proof}
  Let us first show that a transition $t$ in $T$ is ultimately cyclic if, and only if, there a function $\psi\colon T\rightarrow\setQ_{\geq 0}$ (where $\setQ_{\geq 0}$ is the set of non-negative rational numbers) such that $\sum_{t'\in T}\psi(t')\Delta(t')=\vec{0}$ and $\psi(t)>0$. Naturally, if $t$ is ultimately cyclic, then there exists a configuration $\vec{c}$ and a word $w$ in $T^+$ such that $\vec{c}\xrightarrow{w}\vec{c}$. It follows that $\Delta(w)=\vec{0}$ and $t$ occurs in $w$. Let $\psi\colon T\rightarrow\setN$ be the Parikh image of $w$, i.e. $\psi(t')$ is the number of times a transition $t'$ occurs in $w$. Observe that $\Delta(w)=\Delta(\psi)$ and $\psi(t)>0$. Conversely, assume that there is a function $\psi\colon T\rightarrow\setQ_{\geq 0}$ such that $\sum_{t'\in T}\psi(t')\Delta(t')=\vec{0}$ and $\psi(t)>0$. By multiplying $\psi$ by the least common multiple of
 the denominators, we can assume that $\psi$ ranges over the natural numbers. There exists a word $w$ in $T^*$ such that $\psi$ is the Parikh image of $w$. Observe that $\Delta(w)=\Delta(\psi)=\vec{0}$. Now, just observe that there exists a configuration $\vec{c}$ large enough such that $\vec{c}\xrightarrow{w}\vec{c}+\Delta(w)=\vec{c}$. Thus $t$ is ultimately cyclic.

It follows that $U(T)$ is computable in deterministic polynomial time since the membership of a transition $t$ in $U(T)$ reduces to the satisfiability of a linear system of equations over the rational numbers. Moreover, notice that for every $t\in U(T)$ there exists $\psi_t\colon T\rightarrow\setN$ such that $\Delta(\psi_t)=\vec{0}$ and $\psi_t(t)\geq 1$. It follows that $\psi\eqdef\sum_{t\in W(T)}\psi_t$ satisfies the second statement of the theorem.
\end{proof}

\section{Characterization}
\label{sec:theory}
Lemmas~\ref{lem:positive} and~\ref{lem:weak} show that $\Lambda(T)\subseteq M(T)\cap U(T)$. When $M(T)\cap U(T)$ is equal to $T$, the following theorem shows that $\Lambda(T)=T$. The proof of this theorem is inspired by Kosaraju's approach~\cite{Kos1982} for deciding the general reachability problem for Petri nets.
\begin{theorem}\label{thm:main}
  We have $\Lambda(T)=T$ for every Petri net $T$ satisfying $T=M(T)\cap U(T)$.
\end{theorem}
\begin{proof}
  Since $T=M(T)$, Theorem~\ref{thm:positive} shows that there exist two words $w_+,w_-$ in $T^*$ and two configurations $\vec{x},\vec{y}$ such that $\pos{\vec{x}}=I(T)=\pos{\vec{y}}$ and such that:
  $$\vec{y}\xrightarrow{w_-}\vec{0}\xrightarrow{w_+}\vec{x}\;.$$
  Fact~\ref{fact:monotony} shows that for every $n\in\setN$, we have:
  $$n\vec{y}\xrightarrow{w_-^n}\vec{0}\xrightarrow{w_+^n}n\vec{x}\;.$$
  We denote by $\psi_+$ and $\psi_-$ the Parikh image of $w_+$ and $w_-$, resp. Since $T=U(T)$, Theorem~\ref{thm:weak} shows that there exists $\psi_0\colon T\to \setN\moins\{0\}$ such that $\Delta(\psi_0)=\vec{0}$. By replacing $\psi_0$ by $n\psi_0$ with $n\geq 1$ large enough, we can assume without loss of generality that $\psi_0(t)\geq \psi_+(t)+\psi_-(t)$ for every $t\in T$. Let us consider the function $\psi\colon T\to\setN$ satisfying $\psi_+(t)+\psi_-(t)+\psi(t)=\psi_0(t)$ for every $t\in T$. Choose any word $w$ in $T^*$ whose Parikh image is $\psi$. Then we have:
  \begin{eqnarray*}
    \Delta(\psi_+)&=&\Delta(w_+)=\vec{x}\\
    \Delta(\psi_-)&=&\Delta(w_-)=-\vec{y}\\
    \Delta(\psi)&=&\Delta(w)\\
    \Delta(\psi_0)&=&\vec{0}
  \end{eqnarray*}
  We derive from $\psi_++\psi_-+\psi=\psi_0$ the equality $\Delta(\psi_+)+\Delta(\psi_-)+\Delta(\psi)=\Delta(\psi_0)$. It follows that $\vec{x}+\Delta(w)=\vec{y}$. Now, let us consider $\vec{z}\in\{0,1\}^d$ such that $\pos{\vec{z}}=I(T)$. From $\pos{\vec{x}}=I(T)$ it follows that $\vec{x}\geq \vec{z}$. Symmetrically, from $\pos{\vec{y}}=I(T)$ we derive $\vec{y}\geq \vec{z}$. Moreover, since every transition $(\vec{u},\vec{v})\in T$ satisfies $\pos{\vec{u}}\cup\pos{\vec{v}}\subseteq I(T)$, we deduce that there exists $n\geq 1$ large enough such that $n\vec{z}\xrightarrow{w}n\vec{z}+\Delta(w)$. Let us introduce the sequence $\vec{c}_0,\ldots,\vec{c}_n$ of configurations in $\setN^d$ defined by $\vec{c}_j=(n-j)\vec{x}+j\vec{y}$. As $\vec{x},\vec{y}\geq\vec{z}$, we deduce that $\vec{c}_j\geq n\vec{z}$ for every $0\leq j\leq n$. Hence, from $n\vec{z}\xrightarrow{w}n\vec{z}+\Delta(w)$, Fact~\ref{fact:monotony} provides the relation $\vec{c}_{j-1}\xrightarrow{w}\vec{c}_{j-1}+\Delta(w)$. As $\vec{c}_{j-1}+\Delta(w)=\vec{c}_j$, we deduce that $\vec{c}_0\xrightarrow{w^n}\vec{c}_n$. From $\vec{c}_0=n\vec{x}$ and $\vec{c}_k=n\vec{y}$, we obtain:
  $$\vec{0}\xrightarrow{w_+^n w^n w_-^n}\vec{0}$$
  Therefore transitions occurring in $w_+w w_-$ are in $\Lambda(T)$. Notice that the Parikh image of this word is $\psi_0$ which satisfies $\psi_0(t)\geq 1$ for every $t\in T$. Hence $T\subseteq \Lambda(T)$.
\end{proof}

\begin{theorem}\label{thm:Lambda}%
  The set $\Lambda(T)$ is computable in deterministic polynomial time.
\end{theorem}
\begin{proof}
  We associate to every Petri net $T$ the Petri net $\mu(T)\eqdef M(T)\cap U(T)$. Theorems~\ref{thm:positive} and \ref{thm:weak} show that $\mu(T)$ is computable in polynomial time. 
  Lemmas~\ref{lem:positive} and \ref{lem:weak} show that $\Lambda(T)\subseteq \mu(T)$. It follows that $\Lambda(T)=\Lambda(\mu(T))$. In particular, the sequence $T_0,T_1,\ldots$ of Petri nets defined inductively by $T_0=T$ and $T_{n+1}=\mu(T_n)$ ($n\ge 0$) satisfies $\Lambda(T_n)=\Lambda(T)$. Since this sequence is non-increasing for the inclusion relation, there exists $n\leq |T|$ such that $T_{n+1}=T_n$. In that case $\mu(T_n)=T_n$ and Theorem~\ref{thm:main} shows that $\Lambda(T_n)=T_n$. It follows that $\Lambda(T)=T_n$ is computable in deterministic polynomial time.
\end{proof}


Theorem~\ref{thm:Lambda} shows that structural cyclicity can be decided in deterministic polynomial time. In fact, one can easily show that it is, in fact, \ptime-hard as well.\footnote{As usual, \ptime\ denotes the set of all decision problems that can be solved in deterministic polynomial time.}

\begin{theorem}\label{thm:P-hard}%
The structural cyclicity problem is \ptime-hard (under logarithmic space reductions) even with a unary encoding of numbers.
\end{theorem}
\begin{proof}
We prove this theorem by a reduction of the following problem for context-free grammar languages: \emph{Given a context-free grammar $G$, does the language $L(G)$ generated by $G$ contain the empty word $\varepsilon$?} This problem is known to be \ptime-hard (see, e.g.,~\cite[Section~4]{Lange:2011}).

Let $G=(N,\Sigma,P,S)$ be a context-free grammar, where $N$, $\Sigma$, and $P$ are the sets of nonterminals, terminals, and productions, resp., and $S\in N$ is the initial nonterminal. We may assume that $N=\{1,\dots,d\}$, $S=1$, and $\Sigma=\emptyset$. Let $\psi(w)$ be the Parikh image of a word $w\in N^*$. We construct a Petri net $T$ with $d$ dimensions, as follows. $T$ consists of the sub-net $T_0=\{(\psi(l),\psi(r))\mid (l\to r)\in P\}$, and the additional transition $t_0=(\vec 0,\psi(S))=(\vec 0,(1,0,\dots,0))$. Note that the size of $T$ is polynomial in the size of $G$ even under a unary encoding of numbers.

We prove the correctness of the reduction.

Suppose that $S\xrightarrow{P^*} \varepsilon$, where $\xrightarrow{P^*} $ denotes the reflexive and transitive closure of the derivation relation $\xrightarrow P$ of $G$. By construction, for all $u,v\in N^*$, $u\xrightarrow P v$ implies $\psi(u)\xrightarrow{T_0}\psi(v)$. Hence, $\vec 0\xrightarrow{t_0}\psi(S)\xrightarrow{T_0^*}\psi(\varepsilon)=\vec 0$, as required.

For the other direction, let us first note the obvious fact that $\varepsilon\in L(G)$ if, and only if, $S^n\xrightarrow{P^*} \varepsilon$ for some $n>0$. (This follows easily from context-freeness, because every nonterminal appearing in such a derivation is a descendant of a unique one of the $n$ initial occurrences of $S$. Thus, deleting all of them except for the descendants of the first $S$ yields a derivation $S\xrightarrow{P^*} \varepsilon$.) Now, assume that $\vec 0\xrightarrow w\vec 0$ for some $w\in T^+$.  Since $\vec x\xrightarrow{tt_0}\vec y$ implies $\vec x\xrightarrow{t_0t}\vec y$, occurrences of $t_0$ in $w$ can be reordered at the beginning of $w$. So, without loss of generality, we may assume that $w=t_0^n w'$ where $n\geq 0$ and $w'\in T_0^*$. As $\psi(l)>\vec{0}$ for every $(l\to r)\in P$, it follows that $n>0$. Thus, $(n,0,\dots,0)\xrightarrow {w'}\vec 0$. However, the construction of $T_0$ readily implies the following for every word $u\in N^*$: if $\psi(u)\xrightarrow{T_0}\vec y$ for a vector $\vec y$, then there is a word $v\in N^*$ such that $\psi(v)=\vec y$ and $u\xrightarrow P v$. Hence, by induction on $|w'|$ it follows that $S^n\xrightarrow{P^*} v$ for a word $v$ with $\psi(v)=\vec 0$, i.e., we have $S^n\xrightarrow{P^*} \varepsilon$. This completes the proof.
\end{proof}

Combining the above results, we obtain the main result of this paper:

\begin{corollary}
The structural cyclicity problem is \ptime-complete, regardless of whether numbers are encoded in binary or unary. 
\end{corollary}

\section{Conclusion}
\label{sec:conclusion}
In this paper, the structural cyclicity problem has been defined and proved to be decidable in deterministic polynomial time, using a technique inspired by Kosaraju's approach~\cite{Kos1982}. Whereas this approach is non-primitive recursive for deciding the general reachability problem for Petri nets, to our knowledge, this is the first time it is used for deriving a polynomial time algorithm for a Petri net problem.



\bibliographystyle{plain} 
\bibliography{thisbiblio}


\end{document}